\pdfoutput=1
\documentclass[ preliminary, copyright,creativecommons]{eptcs}
\providecommand{\event}{EPTCS: 13th International Conference on Quantum Physics and Logic} 
\usepackage{graphicx}
\usepackage{lipsum}
\usepackage{subfig}
\usepackage{amsmath, amssymb, amsthm, amscd, mathtools}
\usepackage{tikz}
\usepackage{tikz-cd}
\usetikzlibrary{arrows, calc}
\usetikzlibrary{arrows}
\usetikzlibrary{calc}
\usepackage{xcolor}
\usepackage[toc,page]{appendix}
\usepackage{graphicx}

\theoremstyle{definition}
\newtheorem{defn}{Definition}[section]

\theoremstyle{plain}
\newtheorem{theorem}[defn]{Theorem}
\newtheorem{conjecture}[defn]{Conjecture}

\newtheorem{lemma}[defn]{Lemma}
\newtheorem{proposition}[defn]{Proposition}

\theoremstyle{remark}

\theoremstyle{definition}

\usepackage{blindtext}

\newcommand*{\LargerCdot}{\raisebox{-0.8ex}{\scalebox{2}{$\cdot$}}}

\title{On the Cohomology of Contextuality}
\author{Giovanni Car\`{u}
\institute{Quantum Group}
\institute{Department of Computer Science\\
University of Oxford}
\email{giovanni.caru@cs.ox.ac.uk}}
\def\titlerunning{On the Cohomology of Contextuality }
\def\authorrunning{Giovanni Car\`{u}}
\begin{document}
\maketitle

\begin{abstract}
Recent work by Abramsky and Brandenburger used sheaf theory to give a mathematical formulation of non-locality and contextuality. By adopting this viewpoint, it has been possible to define cohomological obstructions to the existence of global sections. In the present work, we illustrate new insights into different aspects of this theory. We shed light on the power of detection of the cohomological obstruction by showing that it is not a complete invariant for strong contextuality even under symmetry and connectedness restrictions on the measurement cover, disproving a previous conjecture. We generalise obstructions to higher cohomology groups and show that they give rise to a refinement of the notion of cohomological contextuality: different ``levels'' of contextuality are organised in a hierarchy of logical implications. Finally, we present an alternative description of the first cohomology group in terms of torsors, resulting in a new interpretation of the cohomological obstructions.

\end{abstract}

\section{Introduction}

Contextuality is one of the most fundamental and peculiar features of quantum mechanics. Inspired by classic no-go theorems by Bell \cite{Bell} and Kochen-Specker \cite{Kochen}, the development of quantum information has been increasingly influenced by the study of this highly non-classical phenomenon. Recent work by Howard et al.\ has even suggested that it actually represents the source of the power of quantum computing \cite{Howard}. The sheaf-theoretic description of non-locality and contextuality introduced in \cite{Abramsky1} has proved that contextuality is in fact a general mathematical property that goes beyond quantum physics and pervades various domains (e.g. relational databases \cite{Abramsky6} and constraint satisfaction \cite{Abramsky5}). 

This rigorous mathematical formulation has allowed the application of powerful methods of \textbf{sheaf cohomology} to the study of the topological structure of contextuality \cite{Abramsky2, Abramsky3}. Central to this approach is the notion of \textbf{cohomological obstruction to the existence of global sections}, i.e. elements of the first \v{C}ech cohomology group that provide a sufficient (but not necessary) condition for the contextuality of empirical models. Although cohomology has been proved to correctly detect contextuality in various well-studied empirical models such as PR boxes \cite{PR}, GHZ states \cite{Greenberger, Greenberger2, Mermin2}, the Peres-Mermin ``magic'' square \cite{Mermin3, Mermin4, Peres} and the whole class of models admitting ``All-vs-Nothing'' arguments \cite{Abramsky2} , there is evidence of a restricted number of false positives (e.g. the Hardy model \cite{Hardy}). 


In the present paper, we illustrate new insights into the properties of cohomological obstructions with the ultimate goal of understanding how such false positives arise. In particular, we aim to give an answer to some of the open questions left by \cite{Abramsky3}: ``Is the cohomological obstruction a full invariant for strong contextuality under suitable restrictions on the measurement scenario?''; ``Can \textbf{higher cohomology groups} be used for the study of contextuality?''; ``Is there a concrete way of describing cohomological obstructions?''.
\newpage
We briefly outline our results:
\begin{itemize}
\item We disprove Conjecture 8.1 of \cite{Abramsky3} by providing an explicit example of a strongly contextual but cohomologically non-contextual empirical model defined on a simple $(2,2,4)$ Bell-type scenario which verifies any reasonably strong form of connectedness and symmetry condition. 
\item We generalise cohomological obstructions to higher cohomology groups. It turns out that this procedure can be done in a natural way only in odd dimensions. We obtain a refinement of the notion of cohomological contextuality: for each $q\ge 0$, we say that a model is $q$-cohomologically contextual if the $q$-th obstruction does not vanish.

\item We show that higher obstructions are organised in a precise hierarchy of logical implications. We also prove that, unfortunately, this result cannot be used in the study of contextuality in no-signalling empirical models. Despite this fact, higher obstructions could potentially be used to study signalling properties.

\item We give a new description of the first cohomology group (thus, in particular, of the cohomology obstructions) using torsors relative to a presheaf.

\end{itemize}

The paper is organised as follows. We summarise the sheaf viewpoint from \cite{Abramsky1} in Section \ref{sec: sheaf}, and recall the main definitions concerning sheaf cohomology in Section \ref{sec: sheaf cohomology}. Section \ref{sec: counterexample} features the counterexample to Conjecture 8.1 of \cite{Abramsky3}. We generalise cohomological obstructions to higher cohomology groups in Section \ref{sec: higher cohomology groups}. Finally, in Section \ref{sec: final}, we present torsors relative to a presheaf, and their relation to the first sheaf cohomology group.

\section{The sheaf-theoretic framework}\label{sec: sheaf}
In this section we recall the main definitions of the sheaf-theoretic approach to non-locality and contextuality \cite{Abramsky1}.

We start by considering a finite discrete space $X$, which can be seen as a set of measurement labels. We define a \textbf{measurement cover} as an antichain $\mathcal{M}=\{C_i\}_{i\in I}$  that satisfies $\bigcup_{i\in I}C_i=X$. This family contains the maximal sets of measurements that can be jointly performed, called \textbf{measurement contexts}. The set $X$, together with the cover $\mathcal{M}$ and a fixed finite \textbf{set of outcomes} $O$, constitute the \textbf{measurement scenario} $\langle X,\mathcal{M}, O\rangle$, which represents the basic setting of the experiment we aim to study.

We consider $X$ as a discrete topological space and define the \textbf{sheaf of events} $$\mathcal{E}:\textbf{Open}(X)^{op}=\mathcal{P}(X)^{op}\rightarrow\textbf{Set},$$ where $\mathcal{E}(U):=O^U$ for each subset $U\subseteq X$, and restriction maps coincide with function restriction: for $U\subseteq U'\subseteq X$, we have
\[
\rho_U^{U'}:=\mathcal{E}(U\subseteq U'):O^{U'}\rightarrow O^U::s\mapsto s\mid_U.
\]
Each $s\in\mathcal{E}(U)$ is called a \textbf{section}, in particular $g\in\mathcal{E}(X)$ is a \textbf{global section}. 

A \textbf{probabilistic empirical model} $e$ is a compatible family $\{e_C\}_{C\in\mathcal{M}}$, where $e_C$ is a probability distribution over $\mathcal{E}(C)$.\footnote{Here, compatibility involves the notion of restriction on distributions which is not defined in this paper since it is not needed (cf. \cite[\S 2.5]{Abramsky1}).} In this paper we will only consider \textbf{possibilistic empirical models}, i.e. the ones generated by the support of a probabilistic model.

%
Such models can be defined as subpresheaves $\mathcal{S}$ of $\mathcal{E}$ that verify the following properties:
\begin{enumerate}
\item\label{cond1} $\mathcal{S}(C)\neq \emptyset$ for all $C\in\mathcal{M}$
\item\label{cond2} $\mathcal{S}$ is \textbf{flasque beneath the cover}, i.e. the map $\mathcal{S}(U\subseteq U')$ is surjective whenever $U\subseteq U'\subseteq C$ for some $C\in\mathcal{M}$. 
\item\label{cond3} Every family $\{s_C\in\mathcal{S}(C)\}_{C\in\mathcal{M}}$ which is compatible (i.e. such that $s_C\mid_{C\cap C'}=s_{C'}\mid_{C\cap C'}$ for all $C, C'\in\mathcal{M}$) induces a global section in $\mathcal{S}(X)$. Note that this global section is unique since $\mathcal{S}$ is a subpresheaf of the sheaf $\mathcal{E}$. 
\end{enumerate}
These conditions state that $\mathcal{S}$ is completely determined by its values $\mathcal{S}(C)$ at each context $C\in\mathcal{M}$: values $\mathcal{S}(U)$ below the cover are fixed by flaccidity, and values for subsets $U$ above the cover are determined by condition \ref{cond3}. Flaccidity beneath the cover can also be interpreted as a possibilistic version of no-signalling. 

Contextuality of an empirical model $\mathcal{S}$ on a measurement scenario $\langle X, \mathcal{M}, O\rangle$ can be characterised as follows
\begin{itemize}
\item Given a context $C\in\mathcal{M}$ and a section $s\in\mathcal{S}(C)$, $\mathcal{S}$ is \textbf{logically contextual at $s$}, or $\mbox{LC}(\mathcal{S}, s)$, if $s$ is not a member of any compatible family. We say that $\mathcal{S}$ is \textbf{logically contextual}, or $\mbox{LC}(\mathcal{S})$, if $\mbox{LC}(\mathcal{S}, s)$ for some possible section $s$. 
\item $\mathcal{S}$ is \textbf{strongly contextual}, or $\mbox{SC}(\mathcal{S})$, if $\mbox{LC}(\mathcal{S}, s)$ for all $s$. In other words there is no global section ($\mathcal{S}(X)=\emptyset$).

\end{itemize}

\subsection{Bundle diagrams}

The structure of the measurement cover can equivalently be described as an abstract simplicial complex having measurements as vertices \cite{Barbosa, Rui}. A set of vertices forms a face whenever the corresponding measurements can be jointly performed, hence contexts correspond to facets of the complex. This viewpoint allows us to graphically represent simple possibilistic empirical models as \textbf{bundle diagrams}. In figure \ref{fig: diagram} we have depicted the bundle diagram of a simple $(2,2,2)$ Bell-type scenario 
involving two agents Alice and Bob who can choose between two binary measurements each ($a_1, a_2$ for Alice and $b_1,b_2$ for Bob). 
\begin{figure}[htbp]
\centering
\includegraphics[scale=0.7]{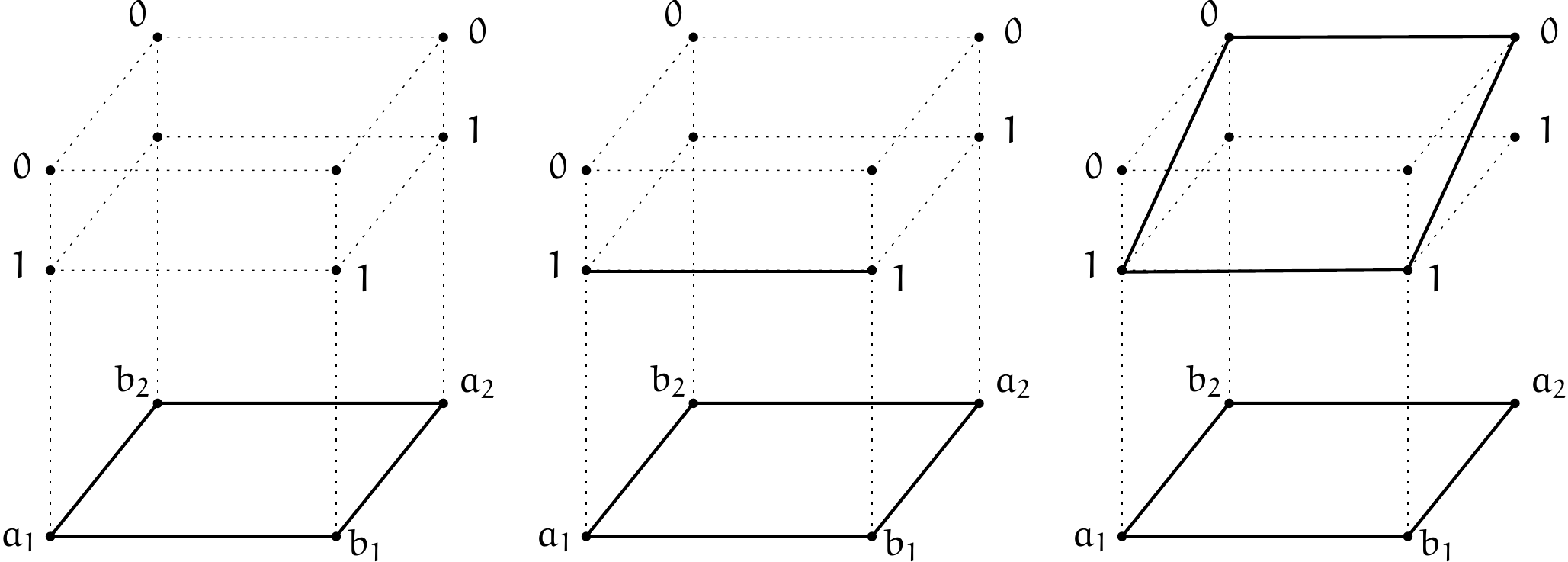}
\caption{A $(2,2,2)$ Bell-type scenario. The section $(a_1,b_1)\mapsto(1,1)$ is represented in the centre. On the right, the global section $(a_1, b_1, a_2,b_2)\mapsto (1,1,0,0)$}\label{fig: diagram}
\end{figure}
The measurement simplicial complex lies at the base of the bundle, and above each vertex is a fibre of the values that can be assigned to each measurement ($0$ or $1$ in this case).  A possible section is represented by an edge connecting the outcomes involved above the corresponding context as in the central diagram of the figure. 
No-signalling corresponds to the property that each edge above a context can be extended to at least one edge above each adjacent context. A global section is represented by a closed path traversing all the fibers exactly once as shown on the right-hand side of Figure \ref{fig: diagram}.

Using this handy representation, we can have an immediate feedback on the contextuality of empirical models. As an example, consider the Hardy model \cite{Hardy, Hardy2} and the Popescu-Rohrlich (PR) box model \cite{PR} represented in Figure \ref{fig: Hardy PR}.

\begin{table}[htbp]
\centering
\subfloat[Hardy model]{
\begin{tabular}{c c | c c c c}
\hline
$A$ & $B$ & $(0,0)$ & $(1,0)$ & $(0,1)$ & $(1,1)$\\
\hline
$a_1$ & $b_1$ & $1$ & $1$ & $1$ & $1$\\
$a_1$ & $b_2$ & $0$ &  $1$ &  $1$ &  $1$\\
$a_2$ & $b_1$ &  $0$ &  $1$ &  $1$ &  $1$\\
$a_2$ & $b_2$ &  $1$ &  $1$ &  $1$ &  $0$\\
\end{tabular}}\quad\quad
\subfloat[PR-Box model]{
\begin{tabular}{c c | c c c c}
\hline
$A$ & $B$ & $(0,0)$ & $(1,0)$ & $(0,1)$ & $(1,1)$\\
\hline
$a_1$ & $b_1$ & $1$ & $0$ & $0$ & $1$\\
$a_1$ & $b_2$ & $1$ &  $0$ &  $0$ &  $1$\\
$a_2$ & $b_1$ &  $1$ &  $0$ &  $0$ &  $1$\\
$a_2$ & $b_2$ &  $0$ &  $1$ &  $1$ &  $0$\\
\end{tabular}}
\end{table}
We can clearly see that the section $(a_1,b_1)\mapsto (0,0)$ in the Hardy bundle, marked in red, is not part of any compatible family, hence the model is logically contextual. However, it is not strongly contextual since there is a global section $(a_,b_1, a_2, b_2)\mapsto (1,1,0,0)$ (marked in blue). On the other hand, all the sections in the PR-Box bundle are not part of any compatible family, which means that the model is strongly contextual. 

\begin{figure}[htbp]
\centering
\includegraphics[scale=0.56]{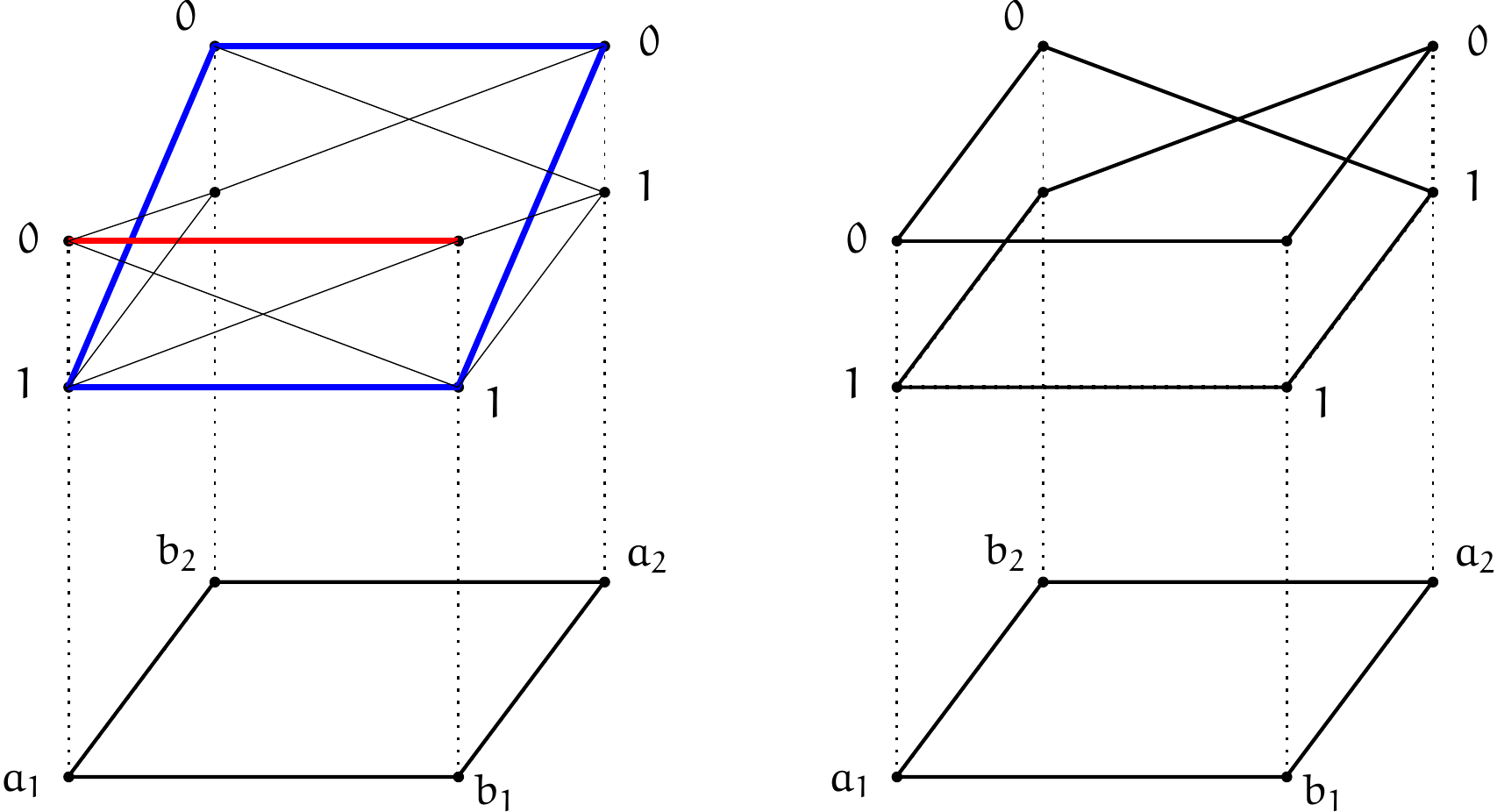}
\caption{The Hardy model and the PR-Box model as bundle diagrams.}\label{fig: Hardy PR}
\end{figure}
\section{Sheaf cohomology}\label{sec: sheaf cohomology}
We recall the main results of \cite{Abramsky2, Abramsky3} concerning cohomological obstructions to the existence of global sections.

Consider a measurement scenario $\langle X, \mathcal{M}, O\rangle$ and an empirical model $\mathcal{S}$ defined on it. We define a presheaf of abelian groups $\mathcal{F}:\mathcal{P}(X)^{op}\rightarrow\textbf{AbGrp}$ that \emph{represents} $\mathcal{S}$. Explicitly, this means that $\mathcal{F}$ verifies conditions \ref{cond1}, \ref{cond2} and \ref{cond3}, and that there is an injection $i:\mathcal{S}\hookrightarrow \mathcal{F}$ such that $i_C(s_C)\neq 0\in\mathcal{F}(C)$ for all $C\in\mathcal{M}$ and for each $s_C\in\mathcal{S}(C)$. Typically, $\mathcal{F}:=F_\mathbb{Z}\mathcal{S}$ is used, where $F_\mathbb{Z}:\textbf{Set}\rightarrow \textbf{AbGrp}$ is the functor that assigns to a set $X$ the free abelian group $F_\mathbb{Z}(X)$ generated by it.\footnote{More generally, the functor $F_R$ can be used for any ring $R$.}

A $q$-\textbf{simplex of the nerve} of $\mathcal{M}$ is a tuple $\sigma=(C_0,\dots, C_q)$ of elements of $\mathcal{M}$ such that $|\sigma|:=\cap_{i=0}^qC_i\neq\emptyset$. The set of $q$-simplices is denoted by $\mathcal{N}(\mathcal{M})^q$. The \textbf{nerve} $\mathcal{N}(\mathcal{M})$ is the abstract simplicial complex generated by all the $\mathcal{N}(\mathcal{M})^q$s. For all $q\ge0$ and each $0\leq j\leq q$, we can define the maps $\partial_j:\mathcal{N}(\mathcal{M})^{q+1}\rightarrow\mathcal{N}(\mathcal{M})^q$ by the expression 
\[
\partial_j(C_0,\dots, C_{q+1}):=(C_0,\dots, C_{j-1}, \hat{C_j}, C_{j+1},\dots, C_{q+1}).
\]
This allows us to define the \textbf{augmented \v{C}ech cochain complex}
\[
0\xrightarrow{~0~} C^0(\mathcal{M}, \mathcal{F})\xrightarrow{\delta^0}C^1(\mathcal{M}, \mathcal{F})\xrightarrow{\delta^1}\dots
\]
where, for all $q\ge 0$,
\[
C^q(\mathcal{M},\mathcal{F}):=\bigoplus_{\sigma\in\mathcal{N}(\mathcal{M})^q}\mathcal{F}(|\sigma|)
\]
is the abelian group of $q$-\textbf{cochains}, and $\delta^q:C^q(\mathcal{M},\mathcal{F})\rightarrow C^{q+1}(\mathcal{M},\mathcal{F})$ defined by
\[
\delta^q(\omega)(\sigma):=\sum_{j=0}^{q+1}(-1)^j\rho_{|\sigma|}^{|\partial_j\sigma|}(\omega(\partial_j\sigma))~~\forall\omega\in C^q(\mathcal{M},\mathcal{F}),~\forall\sigma\in\mathcal{N}(\mathcal{M})^q
\]
is the $q$-\textbf{th coboundary map}, where $\rho_{|\sigma|}^{|\partial_j\sigma|}$ denotes the restriction homomorphism $\mathcal{F}(|\sigma|\subseteq|\partial_j\sigma|)$. \textbf{\v{C}ech cohomology} $\check{H}^\ast(\mathcal{M},\mathcal{F})$ is defined as the cohomology of this augmented cochain complex.

We assume that $\mathcal{M}$ is a \textbf{connected cover}, i.e.\ given $C,C'\in\mathcal{M}$ there exists a sequence of contexts $C=C_0,C_1,\dots, C_n=C'$ such that $C_i\cap C_{i+1}\neq\emptyset$.\footnote{From now on, all the covers will be assumed to be connected.} Thanks to this assumption, cocycles in $Z^0(\mathcal{M},\mathcal{F})\cong\check{H}^0(\mathcal{M},\mathcal{F})$ correspond to compatible families $\{r_C\in\mathcal{F}(C)\}_{C\in\mathcal{M}}$ (i.e. such that $r_C\mid_{C\cap C'}=r_{C'}\mid_{C\cap C'}$ for all $C,C'\in\mathcal{M}$).\footnote{Where $r_C\mid_{C\cap C'}$ is an equivalent notation for $\rho^C_{C\cap C'}(r_C)=\mathcal{F}(C\cap C'\subseteq C)(r_C)$.} 

%

In order to study the extendability of a local section at a fixed context $C_0\in\mathcal{M}$, we shall define the \textbf{relative cohomology of} $\mathcal{F}$. We introduce two auxiliary presheaves. Firstly
\[
\mathcal{F}\mid_{C_0}:\textbf{Open}(X)^{op}\rightarrow\textbf{AbGrp}::U\mapsto \mathcal{F}(U\cap C_0).
\]
The restriction to $C_0$ yields an obvious morphism of sheaves $p^{C_0}:\mathcal{F}\Rightarrow\mathcal{F}\mid_{C_0}$ defined by
\[
p^{C_0}_U:\mathcal{F}(U)\rightarrow\mathcal{F}\mid_{C_0}(U)::r\mapsto r\mid_{C_0\cap U}.
\]
Notice that each $p^{C_0}_U$ is surjective since $\mathcal{F}$ is flasque beneath the cover and $U\cap C_0\subseteq C_0\in\mathcal{M}$. 
The second auxiliary functor is defined by $\mathcal{F}_{\tilde{C}_0}(U):=\ker(p_U^{C_0})$. Thus, we have an exact sequence of presheaves
\begin{equation}\label{equ: exact sequence}
\textbf{0}\Longrightarrow\mathcal{F}_{\tilde{C}_0}\Longrightarrow\mathcal{F}\xRightarrow{~p^{C_0}}\mathcal{F}\mid_{C_0},
\end{equation}
which can be lifted to cochains to
\[
0\longrightarrow C^0(\mathcal{M}, \mathcal{F}_{\tilde{C}_0})\xhookrightarrow{~~~~~~~}C^0(\mathcal{M}, \mathcal{F})\xrightarrow{\bigoplus_Cp_C^{C_0}}C^0(\mathcal{F}, \mathcal{F}\mid_{C_0}),\longrightarrow 0,
\]
where exactness on the right is given by surjectivity of all the $p_C^{C_0}$. The map $\delta^0$ can be correstricted to a map $\tilde{\delta^0}:=\delta^0\mid^{Z^1(\mathcal{M},\mathcal{F})}$ whose kernel is $Z^0(\mathcal{M},\mathcal{F})\cong\check{H}^0(\mathcal{M},\mathcal{F})$ and whose cokernel is isomorphic to $\check{H}^1(\mathcal{M},\mathcal{F})$ (the same procedure can be applied to $\mathcal{F}\mid_{C_0}$ and $\mathcal{F}_{\tilde{C}_0}$). Therefore, by applying the snake lemma to
\begin{center}
  \raisebox{-0.5\height}{\includegraphics{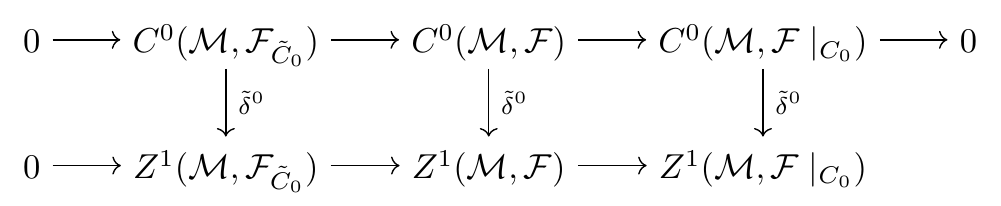}}\\
  \end{center}
%
%
%
we obtain an exact sequence

 \begin{center}
\begin{tikzpicture}[descr/.style={fill=white,inner sep=1.5pt}]
        \matrix (m) [
            matrix of math nodes,
            row sep=2.5em,
            column sep=2.5em,
            text height=1.5ex, text depth=0.25ex
        ]
        { \check{H}^0(\mathcal{M},\mathcal{F}_{\tilde{C}_0}) & \check{H}^0(\mathcal{M},\mathcal{F}) & \check{H}^0(\mathcal{M},\mathcal{F}\mid_{C_0}) \\
        \check{H}^1(\mathcal{M},\mathcal{F}_{\tilde{C}_0}) & \check{H}^1(\mathcal{M},\mathcal{F}) & \check{H}^1(\mathcal{M},\mathcal{F}\mid_{C_0})\\
        };

        \path[overlay,->, font=\scriptsize,>=latex]
        (m-1-1) edge node[auto] {} (m-1-2)
        (m-1-2) edge node[auto] {} (m-1-3)
        (m-1-3) edge[out=355,in=175] node[descr,yshift=0.3ex] {$\gamma_{C_0}$} (m-2-1)
        (m-2-1) edge node[auto] {} (m-2-2)
        (m-2-2) edge node[auto] {} (m-2-3);
\end{tikzpicture}
\end{center}
where the ``snake'' homomorphism $\gamma_{C_0}$ is called the \textbf{connecting homomorphism} relative to the context $C_0$. We have $\mathcal{F}(C_0)\cong\check{H}^0(\mathcal{M}, \mathcal{F}\mid_{C_0})$ 
via the isomorphism
\begin{equation}\label{equ: homomorphism}
\psi^0:\mathcal{F}(C_0)\rightarrow Z^0(\mathcal{M}, \mathcal{F}\mid_{C_0}):: r_{C_0}\mapsto (r_{C_0}\mid_{C_0\cap C})_{C\in\mathcal{M}}.
\end{equation}
Thus, given an element $r_0\in\mathcal{F}(C_0)$, it makes sense to define the \textbf{cohomological obstruction} of $r_0$ as the element $\gamma_{C_0}(r_0)\in\check{H}^1(\mathcal{M},\mathcal{F}_{\tilde{C}_0})$. 

We have the following key result from \cite{Abramsky3}:

\begin{proposition}\label{prop: main}
Let $\mathcal{M}$ be a connected cover, $C_0\in\mathcal{M}$ and $r_0\in\mathcal{F}(C_0)$. Then, $\gamma_{C_0}(r_0)=0$ if and only if there exists a compatible family $\{r_C\in\mathcal{F}(C)\}_{C\in\mathcal{M}}$ such that $r_{C_0}=r_0$. 
\end{proposition}


Given an empirical model $\mathcal{S}$ and a local section $s_0\in\mathcal{S}(C_0)$, we define the following notions
\begin{itemize}
\item $\mathcal{S}$ is \textbf{cohomologically logically contextual at $s_0$}, or $\mbox{CLC}(\mathcal{S}, s_0)$, if $\gamma_{C_0}(s_0)\neq 0$. We say that $\mathcal{S}$ is \textbf{cohomologically logically contextual}, or $\mbox{CLC}(\mathcal{S})$, if $\mbox{CLC}(\mathcal{S}, s)$ for some section $s$. 
\item $\mathcal{S}$ is \textbf{cohomologically strongly contextual}, or $\mbox{\mbox{CSC}}(\mathcal{S})$, if $\mbox{CLC}(\mathcal{S}, s)$ for all sections $s$.
\end{itemize}

The main result of \cite{Abramsky3} provides a sufficient condition for an empirical model to be contextual:
\begin{theorem}\label{thm: main}
Let $\mathcal{S}$ be an empirical model. We have $\rm{CLC}(\mathcal{S})\Rightarrow \mbox{LC}(\mathcal{S})$ and $
\mbox{\mbox{CSC}}(\mathcal{S})\Rightarrow \mbox{SC}(\mathcal{S})$.
\end{theorem} 
However, it is sufficient to consider the Hardy model to show that this condition is not necessary \cite{Abramsky3}. 
Bundle diagrams can be used again to understand how such false positives arise. As an example, consider the diagram of the Hardy model in Figure \ref{fig: Hardy2}.
\begin{figure}[htbp]
\centering
\includegraphics[scale=0.56]{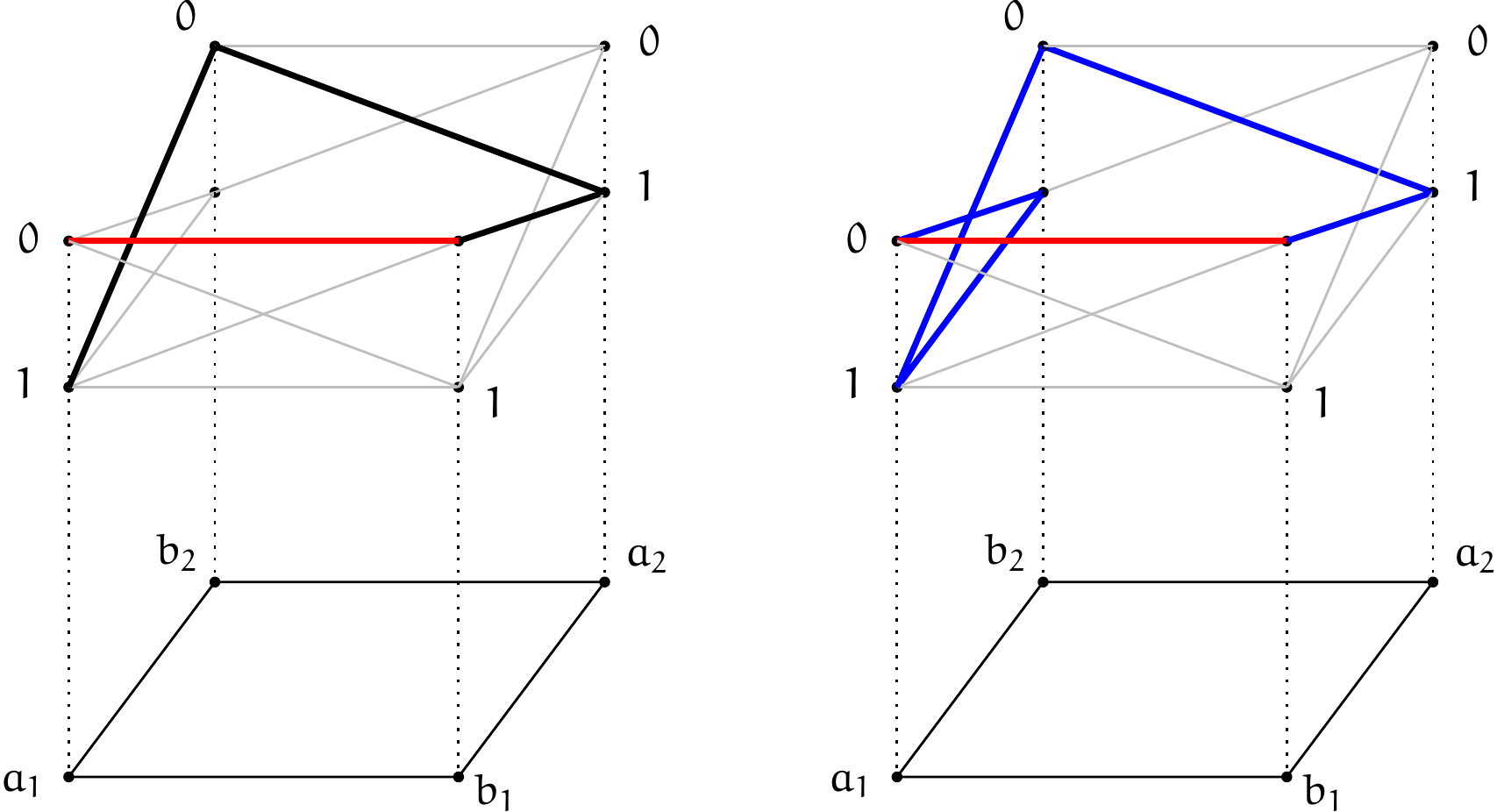}
\caption{The Hardy model is \mbox{LC} but not \mbox{CLC}}\label{fig: Hardy2}
\end{figure}
On the left-hand side we graphically show that the red section $s:=(a_1,b_1)\mapsto (0,0)$ cannot be extended to a compatible family, proving local contextuality at $s$. However, when considering the presehaf of abelian groups $\mathcal{F}:=F_\mathbb{Z}\mathcal{S}$ needed to define cohomological obstructions, we are allowed to take formal linear combinations of sections over the same context. Thus, it is possible to obtain closed paths like the one in blue on the right-hand side of Figure \ref{fig: Hardy2}, which is explicitly defined by
\[
\{s, (a_2,b_1)\mapsto (1,0), (a_2, b_2)\mapsto(1,0), [(a_1,b_2)\mapsto (1,0)]-[(a_1,b_2)\mapsto (1,1)]+[(a_1,b_2)\mapsto (0,1)]\},
\]
and represents a compatible family for $\mathcal{F}$.
The existence of such a closed path shows why cohomology cannot detect contextuality in this case.
\section{A false positive for strong contextuality}\label{sec: counterexample}

In \cite{Abramsky3}, it is brought to attention that, although cohomology can fail to detect logical contextuality as in the case of the Hardy model, it is rather difficult to construct a strongly contextual false positive. Indeed, cohomology is able to detect the strong contextuality of a variety of well-known models, including GHZ states \cite{Greenberger, Greenberger2, Mermin2}, PR Boxes \cite{PR}, the Peres-Mermin ``magic'' square \cite{Mermin3, Mermin4, Peres}, all $\neg GCD$ models \cite{Abramsky3}, and the whole class of models admitting All-vs-Nothing arguments \cite{Abramsky2}. The only known example of a strongly contextual false positive is the Kochen-Specker model \cite{Kochen, ShaneRui} for the cover
\[
\{A,B,C\}, \{B, D, E\}, \{C, D, E\}, \{A, D, F\}, \{A, E, G\},
\]
which ``does not satisfy any reasonable criterion for symmetry, nor does it satisfy any strong form of connectedness'' and where ``the existence of measurements belonging to a single context, namely $F$ and $G$, seems to be crucial'' \cite{Abramsky3}. Due to these limitations, the following conjecture was made:
\begin{conjecture}[Conjecture 8.1 of \cite{Abramsky3}]

Under suitable assumptions of symmetry and connectedness of the cover, the cohomology obstruction is a complete invariant for strong contextuality.
\end{conjecture}
In Figure \ref{fig: counterexample}, we introduce a counterexample to this conjecture. 
\begin{figure}[htbp]
\centering
\includegraphics[scale=0.63]{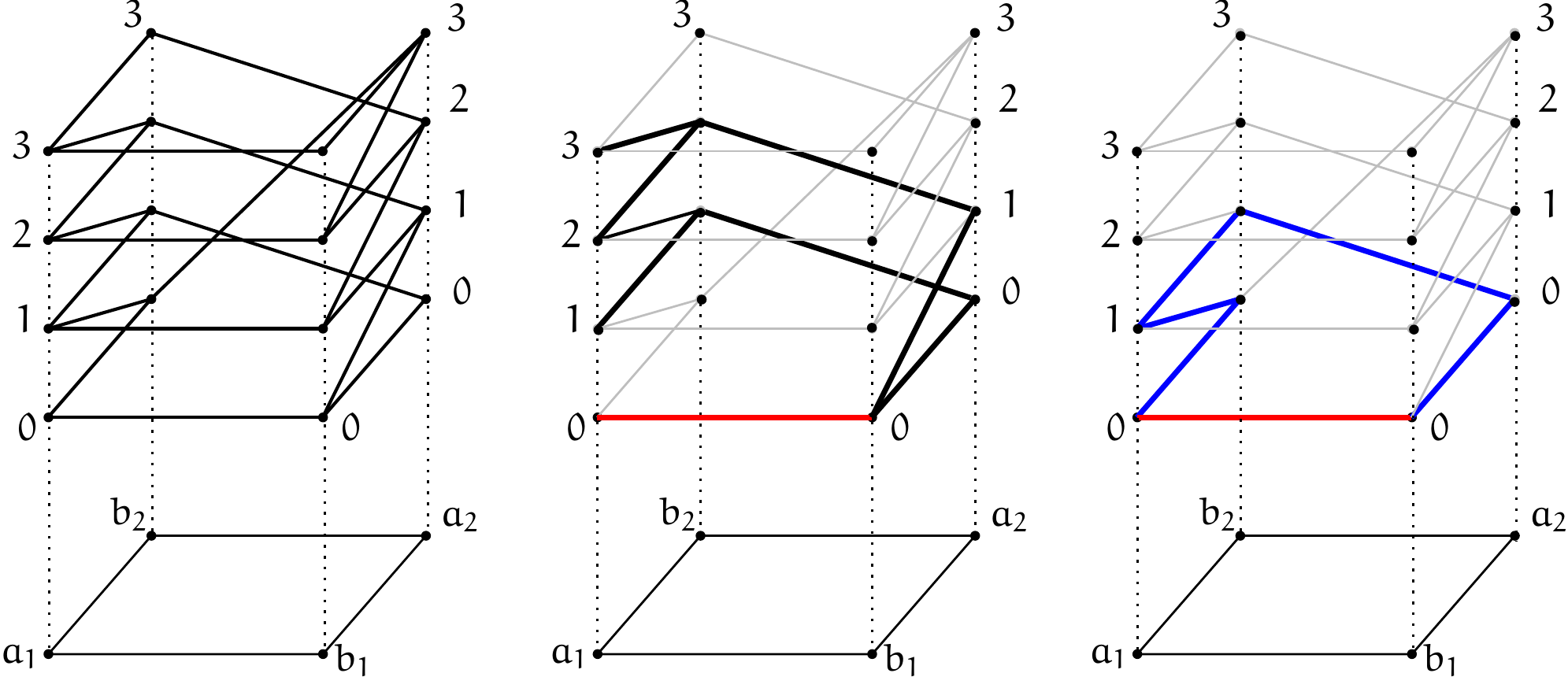}
\caption{A $\mbox{SC}\wedge\neg \mbox{CLC}$ model over a $(2,2,4)$ scenario}\label{fig: counterexample}
\end{figure}

The bundle diagram on the left-hand side represents an empirical model $\mathcal{S}$ (the explicit definition can be found in Appendix \ref{apx: model}) on a $(2,2,4)$ Bell-type scenario. Note that this measurement scenario is extremely simple and verifies any reasonably strong form of symmetry and connectedness.
By carefully analysing the picture, one verifies that none of the sections can be extended to a compatible family of $\mathcal{S}$ (i.e. a closed path containing one and only one section per context), but each one of them is contained in a compatible family of $\mathcal{F}:=F_\mathbb{Z}\mathcal{S}$, namely a closed path similar to the one generating the false positive for the Hardy model (Figure \ref{fig: Hardy2}). As an example, we show this feature by considering the section $s_0:=(a_1,b_1)\mapsto (0,0)$: from the central diagram of Figure \ref{fig: counterexample} it appears clear that this section is non-extendable to a compatible family of $\mathcal{S}$, while the diagram on the right-hand side shows that $s_0$ is part of a compatible family for $\mathcal{F}$, explicitly defined as
\[
\{s_0, (a_2, b_1)\mapsto (0,0), (a_2,b_2)\mapsto (0,1), [(a_1,b_1)\mapsto(1,1)]-[(a_1,b_1)\mapsto(1,0)]+[(a_1,b_1)\mapsto(0,0)]\}.
\]
We conclude that this model is strongly contextual but not cohomologically contextual (not even cohomologically logically contextual), essentially disproving Conjecture 8.1 of \cite{Abramsky3}.
\footnote{The open-endedness of the statement of the conjecture leaves room for a small minority of special cases where cohomology is indeed a full invariant of strong contextuality. An example is given in \cite{Shane}, where it is shown that the conjecture is true for the extremely limited class of symmetric Kochen-Specker models satisfying a condition due to Daykin and H{\"a}ggkvist \cite{Daykin}.}

\section{Higher cohomology groups}\label{sec: higher cohomology groups}

The theory developed so far involves only the first \v{C}ech cohomology group, which contains the obstructions. The existence of badly behaved false positives like the one presented in the previous section motivates a deeper inspection of the higher cohomology groups in search of information on how such extreme cases arise. We will introduce here a generalisation of cohomology obstructions to higher-dimensional cohomology groups. 

Let $\mathcal{F}$ be an abelian presheaf representing an empirical model $\mathcal{S}$ on a scenario $\langle X, \mathcal{M}, O\rangle$ (e.g. $\mathcal{F}:=F_\mathbb{Z}\mathcal{S}$). Let $q\ge 0$ be an integer and fix a context $C_0\in\mathcal{M}$. To each section $s_0\in\mathcal{F}(C_0)$ we associate a $q$-relative cochain $c_{s_0}^q\in C^q(\mathcal{M}, \mathcal{F}\mid_{C_0})$ defined by
\[
c_{s_0}^q(\omega):=s_0\mid_{C_0\cap|\omega|},~~ \forall \omega\in\mathcal{N}(\mathcal{M})^q.
\]
This assignment determines a homomorphism $\psi^q:\mathcal{F}(C_0)\rightarrow C^q(\mathcal{M}, \mathcal{F}\mid_{C_0})$ which generalises the isomorphism \eqref{equ: homomorphism}. Although $\psi^q$ is not an isomorphism in general, it is always injective, which means that different sections in $\mathcal{F}(C_0)$ are mapped to distinct elements of $C^q(\mathcal{M}, \mathcal{F}\mid_{C_0})$.

\begin{lemma}\label{lem: useful2}
For each $q\ge 0$, the homomorphism $\psi^q$ is injective.
\end{lemma}

\begin{proof}
Let $s_0\in\ker(\psi^q)$. Then $c_{s_0}^q=0$, thus in particular $0=c_{s_0}^q(\underbrace{C_0,\dots, C_0}_{q+1\text{~times}})=s_0.$ Therefore $\ker(\psi^q)=0$ and the homomorphism is injective. 
\end{proof}

It turns out that parity in dimension plays an important role:

\begin{lemma}
Let $q\ge 0$. The image of $\psi^q$ is contained in $Z^{q}(\mathcal{M}, \mathcal{F}\mid_{C_0})$ if and only if $q$ is even. 
\end{lemma}

\begin{proof}
Let $s_0\in\mathcal{F}(C_0)$. For any $\omega\in \mathcal{N}(\mathcal{M})^{q+1}$ we have

\[
\delta^{q}\left(c_{s_0}^{q}\right)(\omega)=\sum_{k=0}^{q+1}(-1)^k\rho_{|\omega|}^{|\partial_k \omega|}\left(c_{s_0}^{q}(\partial_k\omega)\right)=\sum_{k=0}^{q+1}(-1)^k\rho_{|\omega|}^{|\partial_k \omega|}\left(s_0\mid_{C_0\cap|\partial_k\omega|}\right)=\sum_{k=0}^{q+1}(-1)^ks_0\mid_{C_0\cap|\omega|}.
\]
The last sum is an alternating sum. Therefore, $\delta^q\left(c_{s_0}^{q}\right)(\omega)=0$ if and only if $q$ is even. 
\end{proof}

%
%
%

Given a $q\ge 0$, we can generalise the construction of the connecting homomorphism to the order $2q$. For each $\sigma\in \mathcal{N}(\mathcal{M})^{2q}$, the exact sequence \eqref{equ: exact sequence} yelds an exact sequence on objects
\begin{equation*}
0\xrightarrow{~~0~~}\mathcal{F}_{\tilde{{C_0}}}(|\sigma|):=ker(p^{C_0}_{|\sigma|})\xrightarrow{~~~~~}\mathcal{F}(|\sigma|)\xrightarrow{~~p^{C_0}_{|\sigma|}~~}\mathcal{F}\mid_{C_0}(|\sigma|)\longrightarrow 0,
\end{equation*}
where exactness on the right is given by flaccidity beneath the cover ($p^{C_0}_{|\sigma|}$ is surjective for all $\sigma$ since $|\sigma|\cap C_0\subseteq C_0$). 
We can sum these morphisms for every $\sigma\in\mathcal{N}(\mathcal{M})^{2q}$ and ``lift'' exactness to the chain level:
\begin{equation}\label{equ: SES2}
0\xrightarrow{~~0~~}C^{2q}(\mathcal{M}, \mathcal{F}_{\tilde{C}_0})\xrightarrow{~~~~~}C^{2q}(\mathcal{M}, \mathcal{F})\xrightarrow{~~\bigoplus_{\sigma}p^{C_0}_{|\sigma|}~~}C^{2q}(\mathcal{M}, \mathcal{F}\mid_{C_0})\longrightarrow 0.
\end{equation}
Then, we take the correstriction $\tilde{\delta}^{2q}$ of the $2q$-th coboundary maps to $Z^{2q+1}$ and obtain

\begin{center}
\begin{tikzpicture}[auto]
\matrix (m) [matrix of math nodes, row sep=2em, column sep=2em, text height=2ex, text depth=0.25ex]
{
0 & C^{2q}(\mathcal{M}, \mathcal{F}_{\tilde{C}_0}) & C^{2q}(\mathcal{M}, \mathcal{F}) & C^{2q}(\mathcal{M}, \mathcal{F}\mid_{C_0}) & 0 \\
0 & Z^{2q+1}(\mathcal{M}, \mathcal{F}_{\tilde{C}_0}) & Z^{2q+1}(\mathcal{M}, \mathcal{F}) & Z^{2q+1}(\mathcal{M}, \mathcal{F}\mid_{C_0}) & \\
};

\path[->, font=\scriptsize]
(m-1-1) edge node[auto] {} (m-1-2)
(m-1-2) edge node[auto] {} (m-1-3)
(m-1-3) edge node[auto] {} (m-1-4)
(m-1-4) edge node[auto] {} (m-1-5)
(m-2-1) edge node[auto] {} (m-2-2)
(m-2-2) edge node[auto] {} (m-2-3)
(m-2-3) edge node[auto] {} (m-2-4)
(m-1-2) edge node[auto] {$\tilde{\delta}^{2q}$} (m-2-2)
(m-1-3) edge node[auto] {$\tilde{\delta}^{2q}$} (m-2-3)
(m-1-4) edge node[auto] {$\tilde{\delta}^{2q}$} (m-2-4);

\end{tikzpicture}
\end{center}
Finally, we apply the snake lemma to this diagram and obtain the $q$-th connecting homomorphism $\tilde{\gamma}^q_{C_0}$.\\

\label{equ: snake2}
  \raisebox{-0.5\height}{\includegraphics{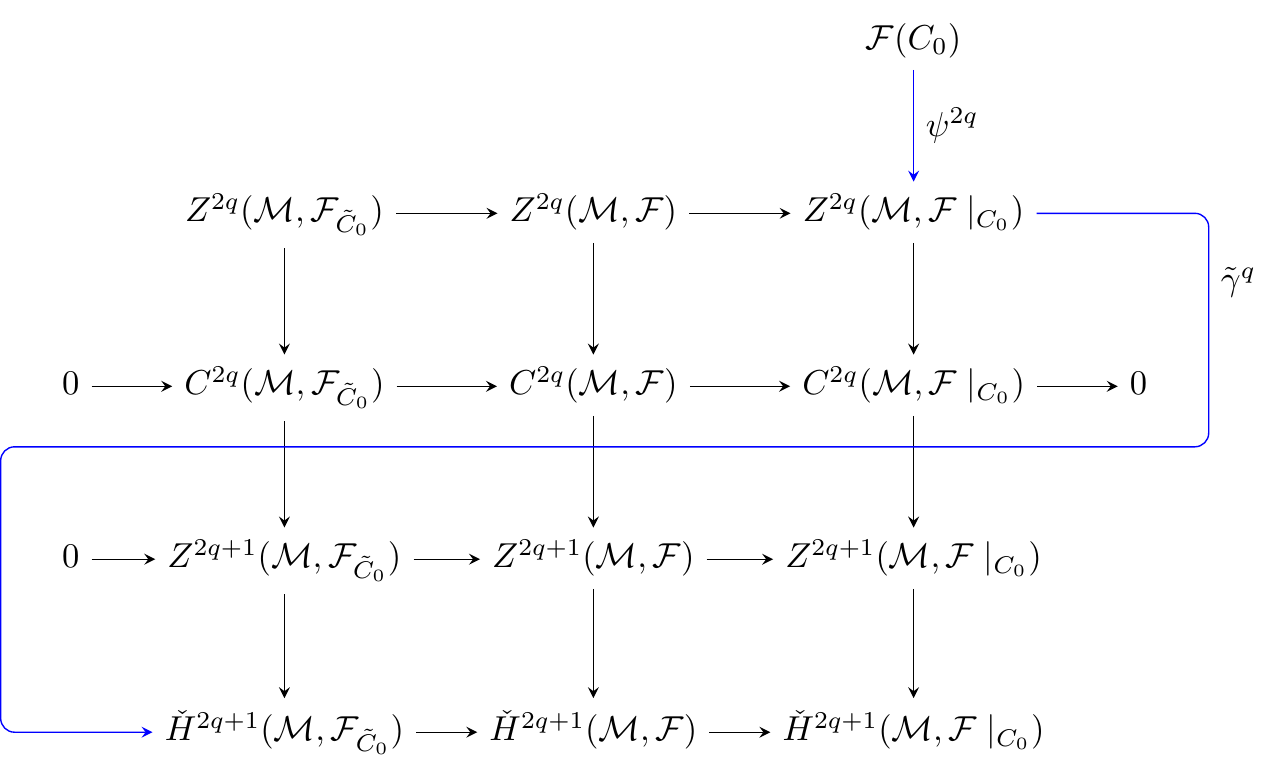}}

\begin{defn}
Let $s_0\in\mathcal{F}(C_0)$. We define the $q$-\textbf{th cohomological obstruction} of $s_0$ as the element 
\[
\gamma^q_{C_0}(s_0):=\tilde{\gamma}^q_{C_0}(\psi^{2q}(s_0))\in \check{H}^{2q+1}(\mathcal{M}, \mathcal{F}).\footnote{Notice that if $q=0$ this definition coincides with the one of cohomological obstruction given before ($\gamma_{C_0}^0=\gamma_{C_0}$).}
\]
The empirical model $\mathcal{S}$ underlying $\mathcal{F}$ is defined to be
\begin{itemize}
 \item \textbf{cohomologically logically} $q$-\textbf{contextual at a section} $s_0$, or $\mbox{CLC}^q(\mathcal{S}, s_0)$, if $\gamma^q_{C_0}(s_0)\neq 0$. We say that $\mathcal{S}$ is \textbf{cohomologically logically} $q$-\textbf{contextual} if $\mbox{CLC}^q(\mathcal{S},s)$ for  some section $s$.

\item \textbf{cohomologically strongly} $q$-\textbf{contextual}, or $\mbox{\mbox{CSC}}^q(\mathcal{S})$, if $\mbox{CLC}^q(\mathcal{S}, s)$ for all $s$.\footnote{Clearly, these definitions depend on the abelian presheaf $\mathcal{F}$ representing $\mathcal{F}$. Typically, $\mathcal{F}:=F_\mathbb{Z}\mathcal{S}$.}
\end{itemize}
\end{defn}
Notice that, due to parity arguments needed to achieve this definition, the cohomological obstruction is generalisable only to odd-dimensional cohomology groups. 

In the case $q=0$, Proposition \ref{prop: main} tells us that the vanishing of the cohomological obstruction is equivalent to the existence of a compatible family in $\mathcal{F}$ containing $s_0$. The analogous result for higher obstructions is the following:

\begin{lemma}\label{lem: general}
Given a $q\ge 0$, a context $C_0\in\mathcal{M}$ and a local section $s_0\in\mathcal{F}(C_0)$, $\gamma_{C_0}^q(s_0)=0$ if and only if there exists a family $s\in Z^{2q}(\mathcal{M},\mathcal{F})$ such that 
\begin{equation}\label{equ: target}
p^{C_0}_{|\sigma|}(s(\sigma))=c^{2q}_{s_0}(\sigma)=s_0\mid_{C_0\cap |\sigma|}~~\forall \sigma\in\mathcal{N}(\mathcal{M})^{2q}.
\end{equation}
\end{lemma}

\begin{proof}
$\gamma^q(s_0)=0\Leftrightarrow\tilde{\gamma}^q(c^{2q}_{s_0})=0\Leftrightarrow c^{2q}_{s_0}\in \ker(\tilde{\gamma}^q)$. Since $\tilde{\gamma}^q$ is defined using the snake lemma, it is part of an exact sequence. Therefore, $c^{2q}_{s_0}\in \ker(\tilde{\gamma}^q)$ if and only if there exists a family $s\in Z^{2q}(\mathcal{M},\mathcal{F})$ such that \eqref{equ: target} is verified.

\end{proof}

\subsection{A hierarchy of cohomological obstructions}

Remarkably, higher cohomology obstructions are organised in a precise hierarchy of implications. 
In the following proposition we show that, if an obstruction vanishes at order $q\ge 0$, it must vanish at any higher order $q'\ge q$ (the proof is given in Appendix \ref{apx: proofs}).

\begin{proposition}\label{thm: hierarchy}
Let $\mathcal{F}$ be an abelian presheaf representing an empirical model $\mathcal{S}$ on a scenario $\langle X, \mathcal{M}, O\rangle$. Let $s_0\in\mathcal{F}(C_0)$. Then $\mbox{CLC}^{q+1}(\mathcal{S}, s_0)\Rightarrow \mbox{CLC}^{q}(\mathcal{S}, s_0)$ for all $q\ge 0$.
\end{proposition}

This result suggests the existence of an infinite number of ``levels'' of contextuality organised in the following hierarchy of logical implications:

\begin{center}
\begin{tikzpicture}[auto]
\matrix (m) [matrix of math nodes, row sep=1.5em, column sep=1.5em, text height=1.5ex, text depth=0.25ex]
{
\mbox{CSC}(\mathcal{S}) & \mbox{CSC}^1(\mathcal{S}) & \dots & \mbox{CSC}^q(\mathcal{S}) & \mbox{CSC}^{q+1}(\mathcal{S})  & \dots\\
\mbox{CLC}(\mathcal{S}) & \mbox{CLC}^1(\mathcal{S}) & \dots & \mbox{CLC}^q(\mathcal{S}) & \mbox{CLC}^{q+1}(\mathcal{S})  & \dots\\
};

\path[-stealth, font=\scriptsize]

(m-1-2) edge[double] node[swap] {\ref{thm: hierarchy}} (m-1-1)
(m-1-3) edge[double] node[auto] {} (m-1-2)
(m-1-4) edge[double] node[auto] {} (m-1-3)
(m-1-5) edge[double] node[swap] {\ref{thm: hierarchy}} (m-1-4)
(m-1-6) edge[double] node[auto] {} (m-1-5)
(m-2-2) edge[double] node[auto] {\ref{thm: hierarchy}} (m-2-1)
(m-2-3) edge[double] node[auto] {} (m-2-2)
(m-2-4) edge[double] node[auto] {} (m-2-3)
(m-2-5) edge[double] node[auto] {\ref{thm: hierarchy}} (m-2-4)
(m-2-6) edge[double] node[auto] {} (m-2-5)

(m-1-1) edge[double] node[auto] {} (m-2-1)
(m-1-2) edge[double] node[auto] {} (m-2-2)
(m-1-4) edge[double] node[auto] {} (m-2-4)
(m-1-5) edge[double] node[auto] {} (m-2-5);

\end{tikzpicture}
\end{center}

However, it turns out that this refinement of the notion of cohomological contextuality cannot be applied to the study of contextuality in no-signalling empirical models (see Appendix \ref{apx: proofs} for the proof):

\begin{proposition}\label{prop: non-contextuality}
No-signalling empirical models are cohomologically $q$-non-contextual for any $q\gneq 0$.
\end{proposition}

Note that, on the other hand, the implications of Proposition \ref{thm: hierarchy} can potentially be used to study the signalling structure of empirical models.\footnote{We thank Kohei Kishida for suggesting this possible application during QPL 2016.} We aim to investigate this aspect in future work.

\section{An alternative description of the first cohomology group}\label{sec: final}

Since higher cohomology groups cannot be used to infer information on how false positives arise, we devote the last section of this paper to a detailed study of the first cohomology group $\check{H}^1(\mathcal{M},\mathcal{F}_{\tilde{C}_0})$. As explained in \cite{Abramsky3}, this group is of crucial importance for the cohomological study of contextuality, as it contains all of the obstructions to the existence of global sections. Its relevance has been also previously highlighted by Penrose in his \textit{On the cohomology of impossible figures} \cite{Penrose}, which presents ``intriguing resemblances'' with our study \cite{Abramsky2}. Yet a full grasp of the nature of its elements is still to be achieved. We propose here a description of $\check{H}^1$ based on the notion of $\mathcal{F}$-torsors, as well as some considerations on the connecting homomorphism $\gamma$. 

\subsection{The connecting homomorphisms $\gamma$}

The first step in understanding cohomological obstructions is studying the connecting homomorphisms.
We present here some insights on how the properties of $\gamma$ can give us information on the type of contextuality of an empirical model. 


\begin{proposition}\label{prop: characterisation}
Let $\mathcal{F}$ be an abelian presheaf representing an empirical model $\mathcal{S}$ on a scenario $\langle X, \mathcal{M}, O\rangle$. The model is cohomologically strongly contextual if and only if $\gamma_C$ is injective for all $C\in\mathcal{M}$. 
\end{proposition}
\begin{proof}
Suppose $\mathcal{S}$ is cohomologically strongly contextual. By Proposition \ref{prop: main}, we have $\gamma_C(s)\neq0$ for all sections $s\in\mathcal{F}(C)$ and all contexts $C\in\mathcal{M}$. In other words, $\ker(\gamma_C)=0$ for all $C\in\mathcal{M}$. For the converse, suppose that $\ker(\gamma_C)=0$ for all $C\in\mathcal{M}$. Then, every non-zero local section of $\mathcal{F}(C)$ has a non-zero cohomological obstruction $\gamma_C(s)$. Thus the model is cohomologically strongly contextual. 
\end{proof}
Thanks to this result, we can give a lower bound for the cardinality of $\check{H}^1(\mathcal{M},\mathcal{F}_{\tilde{C}_0})$ in the case of cohomologically strongly contextual models:
\[
\mbox{\mbox{CSC}}(\mathcal{S})\Rightarrow |\check{H}^1(\mathcal{M},\mathcal{F}_{\tilde{C}_0})|\ge |\mathcal{F}(C_0)|.
\]
On the other hand, given a $\mbox{CLC}\wedge\neg \mbox{\mbox{CSC}}$ model, Proposition \ref{prop: characterisation} implies that two distinct sections may give rise to the same non-zero cohomological obstruction. 

%

The injectivity of a single connecting homomoprhism is a sufficient condition for the strong contextuality of an empirical model.

\begin{proposition}
Let $\mathcal{F}$ be an abelian presheaf representing an empirical model $\mathcal{S}$ on a scenario $\langle X, \mathcal{M}, O\rangle$. If there exists a $C_0\in\mathcal{M}$ such that $\gamma_{C_0}$ is injective, then $\mathcal{S}$ is strongly contextual.
\end{proposition}
\begin{proof}
Suppose there is an injective $\gamma_{C_0}$. If $\mathcal{S}$ is not strongly contextual, there must exist a context $\bar{C}\in\mathcal{M}$ and a section $s\in\mathcal{S}(\bar{C})$ that is extendable to a compatible family $\sigma:=\{s_C\in\mathcal{S}(C)\}_{C\in\mathcal{M}}$. Consider the section $s_{C_0}$ of this family. It is trivially an extendable local section since it is part of the compatible family $\sigma$, thus $\neg \mbox{LC}(\mathcal{S}, s_{C_0})$. By Theorem \ref{thm: main}, this implies $\neg \mbox{CLC}(\mathcal{S}, s_{C_0})$ or, equivalently, $\gamma_{C_0}(s_{C_0})=0$. $\mathcal{F}$ represents $\mathcal{S}$, thus $s_{C_0}$ is non-zero in $\mathcal{F}(C_0)$, hence we conclude that $\ker(\gamma_{C_0})\neq 0$, which means that $\gamma_{C_0}$ is not injective.

\end{proof}
Notice that these two propositions clarify how CSC is a stronger condition than SC: we need all the connecting homomorphisms $\{\gamma_C\}_{C\in \mathcal{M}}$ to be injective in order to conclude that a model is $\mbox{CSC}$, but it is sufficient to have a single injective $\gamma_{C}$ to conclude that it is $\mbox{SC}$.

\subsection{$\mathcal{F}$-torsors and their relation to $\check{H}^1$}
We start by recalling the main definitions; the reader not familiar with the concept of torsor relative to a presheaf might refer to \cite{Alexei} for deeper insights. Let $\mathcal{F}:\textbf{Open}(X)^{op}\rightarrow \textbf{AbGrp}$ be a presheaf of abelian groups over a topological space $X$. An $\mathcal{F}$-\textbf{presheaf} is a presheaf of sets $T$ over $X$ equipped with a morphism of presheaves $\phi:\mathcal{F}\times T\Rightarrow T$ such that, for each open $U\subseteq X$, the map 
\[
\phi_U:\mathcal{F}(U)\times T(U)\rightarrow T(U)::(g, t)\mapsto g\LargerCdot t
\]
is a left action of $\mathcal{F}(U)$ on $T(U)$. Given two $\mathcal{F}$-presheaves $T$ and $T'$, a \textbf{morphism of} $\mathcal{F}$-\textbf{presheaves} from $T$ to $T'$ is a morphism of presheaves $\psi: T\Rightarrow T'$ such that $\psi_U$ is equivariant for all open $U\subseteq X$. An $\mathcal{F}$-presheaf $T$ is called an $\mathcal{F}$-torsor if 
\begin{enumerate}
\item There exists an open cover $\mathcal{V}$ of $X$  that \textbf{trivialises} $T$, i.e. such that $T(V)\neq \emptyset$ for all $V\in\mathcal{V}$.  
\item The action $\phi_U: \mathcal{F}(U)\times T(U)\rightarrow T(U)$ is simply transitive.
\end{enumerate}
The simplest example of $\mathcal{F}$-torsor is the \textbf{trivial} $\mathcal{F}$-\textbf{torsor} $\mathcal{U}\mathcal{F}$,\footnote{Here, $\mathcal{U}:\textbf{AbGrp}\rightarrow \textbf{Set}$ denotes the forgetful functor. To avoid confusion, we will not explicitly show its presence: the trivial $\mathcal{F}$-torsor will be simply denoted by $\mathcal{F}$.} where the action is simply given by $g.\mathcal{U}(h):=\mathcal{U}(g+h)$. We denote by $\mbox{Trs}_\mathcal{F}$ the set of isomorphism classes of $\mathcal{F}$-torsors. It can be proved that an $\mathcal{F}$-torsor $T$ is isomorphic to the trivial $\mathcal{F}$-torsor if and only if $T(X)\neq\emptyset$. 

Now, we adapt this discussion to the case of empirical models. 
Let $\mathcal{F}$ be an abelian presheaf representing an empirical model $\mathcal{S}$ on a scenario $\langle X, \mathcal{M}, O\rangle$, with $\mathcal{M}:=\{C_i\}_{i\in I}$. Let 
\[
\mbox{Trs}(\mathcal{M},\mathcal{F}):=\left\{ T\in \mbox{Trs}_\mathcal{F}\mid T \text{ is trivialised by } \mathcal{M}\right\},
\]
seen as a pointed set with the isomorphism class of the trivial $\mathcal{F}$-torsor as distinguished element. We have the following remarkable result, which is a readaptation of a known correspondence between torsors and cohomology (the proof can be found in Appendix \ref{apx: proofs2}).

\begin{proposition}\label{prop: bijection}
There is a bijection of pointed sets
$\rm Trs(\mathcal{M},\mathcal{F})$ $\cong \check{H}^1(\mathcal{M},\mathcal{F})$.
\end{proposition}

This bijection equips $\mbox{Trs}(\mathcal{M},\mathcal{F})$ with a group structure. The addition of two $\mathcal{F}$-torsors is defined componentwise at each subset $U\subseteq X$ by $g([z])(U)+g([w])(U):=g([z]+[w])(U)$ for all $[z], [w]\in\check{H}^1(\mathcal{M},\mathcal{F})$ (refer to Appendix \ref{apx: proofs2} for the definition of the bijection $g:\check{H}^1(\mathcal{M},\mathcal{F})\rightarrow \mbox{Trs}(\mathcal{M},\mathcal{F})$). Clearly, the above bijection becomes an isomorphism of abelian groups with respect to this addition. 

This results implies that the elements of the first cohomology group $\check{H}^1(\mathcal{M},\mathcal{F}_{\tilde{C}_0})$ relative to a context $C_0\in\mathcal{M}$ (and, in particular, cohomological obstructions) can be seen as isomorphism classes of $\mathcal{F}_{\tilde{C}_0}$-torsors trivialised by the measurement cover $\mathcal{M}$.

Until now, elements of $\check{H}^1$ could only be identified via the abstract equations imposed by the rigid definition of cohomology. The reason why we believe the new description might be more satisfactory, is that despite their seemingly sophisticated definition, torsors are rather simple objects, as explained by Baez in \cite{Baez}. In the simplest terms, an $\mathcal{F}_{\tilde{C}_0}$-torsor is the presheaf $\mathcal{F}_{\tilde{C}_0}$ having lost its identity in each group $\mathcal{F}_{\tilde{C}_0}(U)$, for $U\subseteq X$. Rather than describing the local sections at each $\mathcal{F}_{\tilde{C}_0}(U)$, it measures their difference. We aim to further develop this viewpoint in future work.

\section*{Conclusions}
Sheaf cohomology is a powerful method for the detection of contextuality. However, our work has highlighted some decisive limitations concerning \v{Cech cohomology}. Indeed, it cannot provide a full invariant for contextuality (neither logical nor strong) even under reasonably strong assumptions on symmetry and connectedness of the cover, and, although obstructions can be generalised to higher cohomology groups, they cannot be applied to the study of no-signalling empirical models. 
In future work, we aim to redevelop the sheaf-cohomological study of contextuality from a different viewpoint. The machinery of obstruction theory, a branch of homotopy theory that deals with the extendability of maps, allows the definition of obstructions to the extension of continuous functions in a cohomology theory with coefficients in the homotopy groups. This promising approach will require an adaptation of the concept of empirical model to fit this framework. A possibility would be to formalise the bundle diagram representation and define models as fiber bundles or, more generally, as fibrations. This would allow the definition of Postnikov towers \cite{Postnikov}, which give rise to cohomological obstructions. 

In the last section, we have provided an alternative description of the first cohomology group using torsors relative to a presheaf. Although this approach is still at a developing stage, it allows us to understand cohomological obstructions as a concrete mathematical object. The implications of this new viewpoint will be considered in future work.

\section*{Acknowledgements}
I would like to thank Samson Abramsky for his guidance, Rui Soares Barbosa for the numerous constructive discussions, and Kohei Kishida for his interesting suggestions. Support from the Oxford-Google Deepmind Graduate Scholarship and the
EPSRC Doctoral Training Partnership is also gratefully acknowledged.

\nocite{*}
\bibliographystyle{eptcs}
\bibliography{generic}

\appendix

\section{Explicit definition of the counterexample to conjecture 8.1 of \cite{Abramsky3}}\label{apx: model}
We give the explicit definition of the model introduced in Section \ref{sec: counterexample} as a possibility table:

\begin{table}[htbp]
\centering
\begin{tabular}{cc | cccccccccccccccc}
\hline
$A$ & $B$ & $00$ & $01$ & 10 & 02 & 20 & 03 & 30 & 11 & 12 & 21 & 13 & 31 & 22 & 23 & 32 & 33\\
\hline
$a_1$ & $b_1$ & 1 & 0 & 0 & 0 & 0 & 0 & 0 & 1 & 0 & 0 & 0 & 0 & 1 & 0 & 0 & 1\\
$a_1$ & $b_2$ & 1 & 0 & 1 & 0 & 0 & 0 & 0 & 1 & 0 & 1 & 0 & 0 & 1 & 0 & 1 & 1\\
$a_2$ & $b_1$ & 1 & 0 & 1 & 0 & 0 & 0 & 0 & 1 & 0 & 1 & 0 & 0 & 1 & 0 & 1 & 1\\
$a_2$ & $b_2$ & 0 & 1 & 0 & 0 & 0 & 0 & 1 & 0 & 1 & 0 & 0 & 0 & 0 & 1 & 0 & 0
\end{tabular}
\caption{The possibilistic empirical model pictured in Figure \ref{fig: counterexample}. It is strongly contextual but cohomologically logically non-contextual.}\label{tab: MODEL}
\end{table}

\section{Proofs of the propositions in Section \ref{sec: higher cohomology groups}}\label{apx: proofs}
\begin{proof}[Proof of Proposition \ref{thm: hierarchy}]
We will show the converse: $\neg \mbox{CLC}^q(\mathcal{S}, s_0)\Rightarrow \neg \mbox{CLC}^{q+1}(\mathcal{S}, s_0)$. 
Suppose $\neg\mbox{CLC}^q(\mathcal{S}, s_0)$, then $\gamma_{C_0}^q(s_0)=0$. By Lemma \ref{lem: general} there exists a family $s\in Z^{2q}(\mathcal{M},\mathcal{F})$ such that $$p^{C_0}_{|\sigma|}(s(\sigma))=c^{2q}_{s_0}(\sigma)~\forall\sigma\in\mathcal{N}(\mathcal{M})^{2q}.$$ 
For all $\tau\in \mathcal{N}(\mathcal{M})^{2q+2}$, we define $$f(s)(\tau):= \rho^{|\partial_{2q+1}\partial_{2q+2}\tau|}_{|\tau|}(s(\partial_{2q+1}\partial_{2q+2}\tau))=s(\partial_{2q+1}\partial_{2q+2}\tau)\mid_{|\tau|}.$$

Notice that $f(s)(\tau)\in\mathcal{F}(|\tau|)$, thus $f(s)\in C^{2q+2}(\mathcal{M},\mathcal{F})$. We can actually show that $f(s)$ is in $Z^{2q+2}(\mathcal{M},\mathcal{F})$ as follows.
Given an arbitrary $\nu\in \mathcal{N}(\mathcal{M})^{2q+3}$, we have

\begin{equation}\label{equ: part1}
\begin{split}
\delta^{2q+2}(f(s))(\nu) &= \sum_{k=0}^{2q+3}(-1)^k\rho^{|\partial_k\nu|}_{|\nu|}(f(s)(\partial_k\nu))=\sum_{k=0}^{2q+3}(-1)^k\rho^{|\partial_k\nu|}_{|\nu|}\rho^{|\partial_{2q+1}\partial_{2q+2}\partial_k\nu|}_{|\partial_k\nu|}(s(\partial_{2q+1}\partial_{2q+2}\partial_k\nu))\\
&=\sum_{k=0}^{2q+3}(-1)^k\rho^{|\partial_{2q+1}\partial_{2q+2}\partial_k\nu|}_{|\nu|}(s(\partial_{2q+1}\partial_{2q+2}\partial_k\nu))\\
&= \sum_{k=0}^{2q+1}(-1)^k\rho^{|\partial_{2q+1}\partial_{2q+2}\partial_k\nu|}_{|\nu|}(s(\partial_{2q+1}\partial_{2q+2}\partial_k\nu))+\rho^{|\partial_{2q+1}\partial_{2q+2}\partial_{2q+2}\nu|}_{|\nu|}(s(\partial_{2q+1}\partial_{2q+2}\partial_{2q+2}\nu))\\
&-\rho^{|\partial_{2q+1}\partial_{2q+2}\partial_{2q+3}\nu|}_{|\nu|}(s(\partial_{2q+1}\partial_{2q+2}\partial_{2q+3}\nu))\\
&
\end{split}
\end{equation}
Notice that the last two terms of te sum cancel out since $\partial_{2q+1}\partial_{2q+2}\partial_{2q+2} \nu=\partial_{2q+1}\partial_{2q+2}\partial_{2q+3}\nu$. Hence,
\begin{equation}\label{equ: part2}
\begin{split}
\delta^{2q+2}(f(s))(\nu) &\stackrel{\eqref{equ: part1}}{=}\sum_{k=0}^{2q+1}(-1)^k\rho^{|\partial_{2q+1}\partial_{2q+2}\partial_k\nu|}_{|\nu|}(s(\partial_{2q+1}\partial_{2q+2}\partial_k\nu))\\
&= \sum_{k=0}^{2q+1}(-1)^k\rho^{|\partial_k\partial_{2q+1}\partial_{2q+2}\nu|}_{|\nu|}(s(\partial_k\partial_{2q+1}\partial_{2q+2}\nu)),
\end{split}
\end{equation}
where the last equality is valid since now $0\leq k\leq 2q+1$ and therefore it is unimportant whether we cancel the $k$-th term before or after having canceled the $(2q+2)$-th and the $(2q+1)$-th. 
We can now relabel $\partial_{2q+1}\partial_{2q+2}\nu:=\tilde{\nu}\in\mathcal{N}(\mathcal{M})^{2q+1}$ 
and obtain
\begin{equation*}
\begin{split}
\delta^{2q+2}(f(s))(\nu) &\stackrel{\eqref{equ: part2}}{=} \sum_{k=0}^{2q+1}(-1)^k\rho^{|\partial_k\tilde{\nu}|}_{|\nu|}(s(\partial_k\tilde{\nu}))=\delta^{2q}(s)(\tilde{\nu})=0,
\end{split}
\end{equation*}
where the last equality is due to the fact that $s\in Z^{2q}(\mathcal{M},\mathcal{F})$. 

Let $\sigma\in\mathcal{N}(\mathcal{M})^{2q+2}$. We have
\begin{equation*}
\begin{split}
p^{C_0}_{|\sigma|}(f(s)(\sigma)) &= f(s)(\sigma)\mid_{|\sigma|\cap C_0}= s(\partial_{2q+1}\partial_{2q+2}\sigma)\mid_{|\sigma|\cap C_0}=s(\tilde{\sigma})\mid_{|\sigma|\cap C_0}=s(\tilde{\sigma})\mid_{|\tilde{\sigma}|\cap |\sigma|\cap C_0}\\
&= \left(s(\tilde{\sigma})\mid_{|\tilde{\sigma}|\cap C_0}\right)\mid_{|\sigma|}= \left(p^{C_0}_{|\tilde{\sigma}|}(s(\tilde{\sigma}))\right)\mid_{|\sigma|}=
\left( c^{2q}_{s_0}(\tilde{\sigma})\right)\mid_{|\sigma|}=\left(s_0\mid_{|\tilde{\sigma}|\cap C_0}\right)\mid_{|\sigma|}\\
&= s_0\mid_{|\tilde{\sigma}|\cap|\sigma|\cap C_0}=s_0\mid_{|\sigma|\cap C_0}=c^{2q+2}_{s_0}(\sigma).
\end{split}
\end{equation*}
By Lemma \ref{lem: general} this implies $\gamma_{C_0}^{q+1}(s_0)=0$.
\end{proof}

\begin{proof}[Proof of Proposition \ref{prop: non-contextuality}]
Consider an abelian presheaf $\mathcal{F}$ representing an empirical model $\mathcal{S}$ on a scenario $\langle X, \mathcal{M},O\rangle$, where $\mathcal{M}:=\{C_i\}_{i\in I}$.
Let $C_0\in\mathcal{M}$ be an arbitrary context, and $s_{C_0}\in \mathcal{F}(C_0)$ an arbitrary section. By no-signalling, there exists a family $\{s_{C_i}\in \mathcal{F}(C_i)\}_{i\in I}$ such that $
s_{C_i}\mid_{C_i\cap C_0}= s_{C_0}\mid_{C_i\cap C_0}$ for all $i$. We define $z\in C^2(\mathcal{M},\mathcal{F})$ by the expression
\begin{equation*}
z(\omega):= s_{\partial_0\partial_2\omega}\mid_{|\omega|}~\in \mathcal{F}(|\omega|)~~\forall\omega\in\mathcal{N}(\mathcal{M})^2.
\end{equation*}
More explicitly, given an $\omega:= (C_i, C_j, C_k)\in\mathcal{N}(\mathcal{M})^2$, we define
\begin{equation*}
z(C_i,C_j,C_k):= s_{C_j}\mid_{C_i\cap C_j\cap C_k}~\in\mathcal{F}(C_i\cap C_j\cap C_k).
\end{equation*}
Given a general $\sigma:=(C_i, C_j, C_k, C_l)\in\mathcal{N}(\mathcal{M})^3$, we have

\begin{equation*}
\begin{split}
\delta^2(z)(\sigma)&= z(C_j, C_k, C_l)\mid_{|\sigma|}-z(C_i, C_k, C_l)\mid_{|\sigma|}+z(C_i, C_j, C_l)\mid_{|\sigma|}-z(C_i, C_j, C_k)\mid_{|\sigma|}\\
&= s_{C_k}\mid_{|\sigma|}-s_{C_k}\mid_{|\sigma|} + s_{C_j}\mid_{|\sigma|} - s_{C_j}\mid_{|\sigma|}=0,
\end{split}
\end{equation*}
thus $z\in Z^2(\mathcal{M},\mathcal{F})$. Moreover, for any general $\omega=(C_i, C_j, C_k)\in\mathcal{N}(\mathcal{M})^2$ we have
\begin{equation*}
\begin{split}
p^{C_0}_{|\omega|}(z(\omega))&=z(\omega)\mid_{|\omega|\cap C_0}=s_{C_j}\mid_{C_i\cap C_j\cap C_k\cap C_0}= \left(s_{C_j}\mid_{C_j\cap C_0}\right)\mid_{C_i\cap C_j\cap C_k\cap C_0}= \left( s_{C_0}\mid_{C_j\cap C_0}\right)\mid_{|\omega|\cap C_0}\\
&= s_{C_0}\mid_{|\omega|\cap C_0}= c^2_{s_{C_0}}(\omega).
\end{split}
\end{equation*}
By Lemma \ref{lem: general}, this result implies $\gamma_{C_0}^1(s_{C_0})=0$, and by Proposition \ref{thm: hierarchy}, we conclude $\neg \mbox{CLC}^q(\mathcal{S})$ for all $q>0$. 
\end{proof}

\section{Proof of Proposition \ref{prop: bijection}}\label{apx: proofs2}

\begin{proof}[Proof of Proposition \ref{prop: bijection}]

Let $T\in \mbox{Trs}(\mathcal{M},\mathcal{F})$. We arbitrarily choose a collection $\{t_i\in T(C_i)\}_{i\in I}$\footnote{This is possible since $\mathcal{M}$ trivialises $T$.}. By simple transitivity, for all $i,j\in I$ there exists a unique $g_{ij}\in\mathcal{F}(C_i\cap C_j)$ such that $g_{ij}\LargerCdot t_j\mid_{C_i\cap C_j}=t_i\mid_{C_i\cap C_j}$. We also have

\begin{equation*}
\begin{split}
(g_{kj}\mid_{C_i\cap C_j\cap C_k}+g_{ji}\mid_{C_i\cap C_j\cap C_k})\LargerCdot t_i\mid_{C_i\cap C_j\cap C_k} &=
g_{kj}\mid_{C_i\cap C_j\cap C_k}\LargerCdot\left(g_{ji}\LargerCdot t_i\mid_{C_i\cap C_j}\right)\mid_{C_i\cap C_j\cap C_k}\\
&= g_{kj}\mid_{C_i\cap C_j\cap C_k}\LargerCdot\left(t_j\mid_{C_i\cap C_j}\right)\mid_{C_i\cap C_j\cap C_k}\\
&= \left(g_{kj}\LargerCdot t_j\mid_{C_j\cap C_k}\right)\mid_{C_i\cap C_j\cap C_k}= t_k\mid_{C_i\cap C_j\cap C_k}= g_{ki}\LargerCdot t_i\mid_{C_i\cap C_j\cap C_k},
\end{split}
\end{equation*}
which implies $g_{kj}	\mid_{C_i\cap C_j\cap C_k}+ g_{ji}\mid_{C_i\cap C_j\cap C_k}=g_{ki}\mid_{C_i\cap C_j\cap C_k}$ for all $i,j,k\in I$ by simple transitivity. This equation tells us that $\check{T}$, defined by $\check{T}(C_i, C_j):=g_{ij}$ for all $i,j\in I$, is a $1$-cocycle. Let

\[
f: \mbox{Trs}(\mathcal{M},\mathcal{F}) \rightarrow \check{H}^1(\mathcal{M},\mathcal{F}):: T\mapsto [\check{T}].
\]

In order to show that this map is well-defined, we need to prove that $\check{T}$ is independent of the choice of the family $\{t_i\}_{i\in I}$. Suppose we choose $\{t_i'\in T(C_i)\}_{i\in I}$ instead, then we obtain a family $\{g'_{ij}\in\mathcal{F}(C_i\cap C_j)\}_{i,j\in I}$ as before. By simple transitivity, for each $i\in I$ there exists an element $g_i\in\mathcal{F}(C_i)$ such that $g_i\LargerCdot  t_i'=t_i$. Thus, we obtain a family $g:=\{g_i\in\mathcal{F}(C_i)\}_{i\in I}$. We have
\begin{equation*}
\begin{split}
\left(g_i\mid_{C_i\cap C_j}+g'_{ij}\right) \LargerCdot  t'_j\mid_{C_i\cap C_j} &= g_i\mid_{C_i\cap C_j}\LargerCdot \left(g'_{ij}\LargerCdot  t'_j\mid_{C_i\cap C_j}\right)=g_i\mid_{C_i\cap C_j}\LargerCdot  t_i'\mid_{C_i\cap C_j}=t_i\mid_{C_i\cap C_j},~~\forall i,j\in I.
\end{split}
\end{equation*}
On the other hand, 
\begin{equation*}
\begin{split}
\left(g_{ij}+g_j\mid_{C_i\cap C_j}\right) \LargerCdot  t'_j\mid_{C_i\cap C_j} = g_{ij}\LargerCdot \left(g_j\mid_{C_i\cap C_j}\LargerCdot  t'_j\mid_{C_i\cap C_j}\right)=g_{ij}\LargerCdot  t_j\mid_{C_i\cap C_j}=t_i\mid_{C_i\cap C_j},~~\forall i,j\in I.
\end{split}
\end{equation*}
Again, by simple transitivity, this implies $g_i\mid_{C_i\cap C_j}+g'_{ij}=g_{ij}+g_j\mid_{C_i\cap C_j}$ for all $i,j\in I$, which is equivalent to say $\delta^0(g)(C_i, C_j)=g'_{ij}-g_{ij}$ for all $i,j\in I$. Consequently, it does not matter whether we define $\check{T}(C_i, C_j):= g_{ij}$ or $\check{T}(C_i, C_j):= g'_{ij}$ since these two $1$-cocycles are cohomologous.

Notice that $f$ maps the trivial $\mathcal{F}$-torsor to $0\in\check{H}^1(\mathcal{M},\mathcal{F})$, thus it is a morphism of pointed sets. To prove that $f$ is a bijection, we introduce an inverse $g: \check{H}^1(\mathcal{M},\mathcal{F})\rightarrow \mbox{Trs}(\mathcal{M},\mathcal{F})$. Given $[z]\in\check{H}^1(\mathcal{M},\mathcal{F})$, we define the presheaf $g([z]):\textbf{Open}(X)^{op}\rightarrow \textbf{AbGrp}$ by the expression

\begin{equation*}
g([z])(U):=\left\{ (t_i)_{i\in I}\in \bigoplus_{i\in I}\mathcal{F}(C_i\cap U) \mathrel{\Bigg|} t_i\mid_{C_i\cap C_j\cap U}-t_j\mid_{C_i\cap C_j\cap U}=z(C_i, C_j)\mid_{C_i\cap C_j\cap U},\forall i,j\in I\right\},
\end{equation*}
for any $U\subseteq X$. The restriction maps are given by
$g([z])(U\subseteq U'):: (t'_i)_{i\in I}\mapsto (t_i'\mid_{C_i\cap U})_{i\in I}$.
We define an $\mathcal{F}$-action on $g([z])$ by the expression
$g\LargerCdot (t_i)_{i\in I}:= (t_i-g\mid_{C_i\cap U})_i$, for any $g\in\mathcal{F}(U)$. 

 We need to show that $g([z])\in\mbox{Trs}(\mathcal{M},\mathcal{F})$. To do so, we show that for any context $C_j\in\mathcal{M}$, there exists an isomorphism of $\mathcal{F}\mid_{C_j}$-presheaves $\mathcal{F}\mid_{C_j}\Rightarrow g([z])\mid_{C_j}$ (recall that $\mathcal{F}$ denotes the trivial $\mathcal{F}$-torsor). Consider a $U\subseteq C_j$. The map

 \[h^j_U: \mathcal{F}\mid_{C_j}(U)  \rightarrow  g([z])\mid_{C_j}(U):: g \mapsto  \left(z(C_i, C_j)\mid_{C_i\cap C_j \cap U}-g\mid_{C_i\cap C_j\cap U}\right)_{i\in I}\]
is an isomorphism with inverse

\[
k^j_U: g([z])\mid_{C_j}(U) \rightarrow \mathcal{F}\mid_{C_j}(U) :: (t_i)_{i\in I} \mapsto -t_j.
\]
In fact, $h_U^j$ is equivariant since
\begin{equation*}
\begin{split}
g\LargerCdot h^j_U(h) &=g\LargerCdot \left(z(C_i, C_j)\mid_{C_i\cap C_j \cap U}-h\mid_{C_i\cap C_j\cap U}\right)_{i\in I}=\left( z(C_i, C_j)\mid_{C_i\cap C_j \cap U}-h\mid_{C_i\cap C_j\cap U}-g\mid_{C_i\cap C_j\cap U}\right) \\
&= h_U^j(\mathcal{U}(g+h))=h_U^j(g\LargerCdot h),
\end{split}
\end{equation*}
where the last action is the one of the trivial $\mathcal{F}$-torsor.
Moreover, $k_U^j$ is indeed the inverse of $h_U^j$:
\begin{equation*}
\begin{split}
h^j_U\left(k^j_U\left((t_i)_{i\in I}\right)\right) &= h_U(-t_j)=\left(z(C_i, C_j)\mid_{C_i\cap C_j \cap U}+t_j\right)_{i\in I}=\left(t_i-t_j+t_j\right)_{i\in I}=(t_i)_{i\in I},
\end{split}
\end{equation*}
and
\begin{equation*}
\begin{split}
k^j_U(h^j_U(g)) =k_U\left(\left(z(C_i, C_j)\mid_{C_i\cap C_j \cap U}-g\mid_{C_i\cap C_j\cap U}\right)_{i\in I}\right)= -z(C_j, C_j)\mid_{C_j\cap U}+g=g,
\end{split}
\end{equation*}
where the last equality is due to the fact that $z$ is a $1$-cocycle.
Since $\mathcal{F}\mid_{C_j}\cong g([z])\mid_{C_j}$ for all contexts $C_j$, we now that $g([z])$ is an $\mathcal{F}$-torsor trivialised by the measurement cover $\mathcal{M}$.

We also need to show that the definition of $g$ is independent of the choice of the representative $z$ of the $1$-cocycle $[z]$. Suppose we take a cohomologous $1$-cocycle $z'$. Then there exists a family $h:=\{h_i\in\mathcal{F}(C_i)\}_{i\in I}$ such that $z'(C_i, C_j)-z(C_i, C_j)=\delta^0(h)$. 
Then we can define an isomorphism of $\mathcal{F}$-torsors $g([z])\cong g([z'])$ induced by the maps
\begin{equation*}
\psi_U: g([z])(U)\rightarrow g([z'])(U):: (t_i)_{i\in I}\mapsto (h_i\mid_{C_i\cap U}+t_i)_{i\in I}.
\end{equation*}
In fact, this map is equivariant since
\begin{equation*}
\begin{split}
g\LargerCdot\psi_U\left((t_i)_{i\in I}\right) &=g\LargerCdot\left((h_i\mid_{C_i\cap U}+t_i)_{i\in I}\right)= \left(h_i\mid_{C_i\cap U}+t_i-g\mid_{C_i\cap U}\right)_{i\in I}\\
&= \psi_U\left(\left(t_i-g\mid_{C_i\cap U}\right)_{i\in I}\right)=\psi_U\left(g\LargerCdot(t_i)_{i\in I}\right),
\end{split}
\end{equation*}
and its inverse is clearly
\begin{equation*}
g([z'])(U)\rightarrow g([z])(U):: (t'_i)_{i\in I}\mapsto (t'_i-h_i\mid_{C_i\cap U})_{i\in I}.
\end{equation*}

We can finally show that $g$ is the inverse of $f$. 
\begin{itemize}

\item Let $T\in\mbox{Trs}(\mathcal{M},\mathcal{F})$. We want to show that $T\cong g([\check{T}])$. Let $U\subseteq X$, and suppose that $\check{T}$ is defined with respect to the family $\{t_i\in T(C_i)\}_{i\in I}$. Consider an element $s\in T(U)$ and the induced family $\{s_i\in T(U\cap C_i)\}_{i\in I}:= \{s\mid_{C_i\cap U}\}_{i\in I}$.\footnote{Note the similarities with the construction of cohomological obstruction in \cite{Abramsky3}, where we take a no-signalling family for the initial section.} By simple transitivity, for each $i\in I$ there is a unique $g_i\in\mathcal{F}(C_i\cap U)$ such that $g_i\LargerCdot s_i=t_i\mid_{C_i\cap U}$. This allows us to define the isomorphism
\[
\phi_U: T(U)\rightarrow g([\check{T}])(U):: s\rightarrow (g_i)_{i\in I}
\]
We leave to the reader the rather simple verification of the fact it is actually an isomorphism, but we explicitly show that it is equivariant. To see this, let $h\in\mathcal{F}(U)$. We have $\phi_U(h\LargerCdot s)=(k_i)_{i\in I}$, where, for all $i\in I$, $k_i$ is the unique element in $\mathcal{F}(C_i\cap U)$ such that $k_i\LargerCdot (h\LargerCdot s)\mid_{C_i\cap U}=t_i\mid_{C_i\cap U}$. More explicitly, $k_i$ is the unique element such that 
\[
k_i\LargerCdot(h\mid_{C_i\cap U}\LargerCdot s_i)=t_i\mid_{C_i\cap U},
\]
which is equivalent to
\[
(k_i+h\mid_{C_i\cap U})\LargerCdot s_i=t_i\mid_{C_i\cap U}.
\]
On the other hand, $h\LargerCdot \phi_U(s)=h\LargerCdot (g_i)_{i\in I}=(g_i-h\mid_{C_i\cap U})_{i\in I}$. Since

\[
(g_i-h\mid_{U\cap C_i})\LargerCdot(h\mid_{C_i\cap U}\LargerCdot s_i)= (g_i-h\mid_{C_i\cap U}+h\mid_{C_i\cap U})\LargerCdot s_i=g_i\LargerCdot s_i=t_i\mid_{C_i\cap U},
\]
we conclude that by simple transitivity that $k_i=g_i-h\mid_{U\cap C_i}$ for all $i\in I$, which leads to $h\LargerCdot \phi_U(s)=\phi_U(h\LargerCdot s)$. 
\item Let $[z]\in\check{H}^1(\mathcal{M},\mathcal{F})$. We want to show that $f(g([z]))=[z]$. We construct the family $\{t_k\in g([z])(C_k)\}$ given by $t_k:=(z(C_i, C_k))_{i\in I}$ and we use it to define $f(g([z]))$ by setting, for all $i,j\in I$, $f(g([z]))(C_i, C_j)$ to be the unique element $g_{ij}\in\mathcal{F}(C_i\cap C_j)$ such that $g_{ij}\LargerCdot t_j\mid_{C_i\cap C_j}=t_i\mid_{C_i\cap C_j}$.
Notice that 
\begin{equation*}
\begin{split}
z(C_l, C_k)\LargerCdot t_k\mid_{C_l\cap C_k} &=z(C_l, C_k)\LargerCdot (z(C_i, C_k)\mid_{C_i\cap C_l\cap C_k})_{i\in I}=(z(C_i, C_k)\mid_{C_i\cap C_l\cap C_k}-z(C_l, C_k)\mid_{C_i\cap C_l\cap C_k})_{i\in I}\\
&=(z(C_i, C_l))_{i\in I}=t_l\mid_{C_l\cap C_k}.
\end{split}
\end{equation*}
Therefore, by simple transitivity, $g_{ij}=z(C_i, C_j)$ for all $i,j\in I$, proving $f(g([z]))=[z]$.

\end{itemize}
\end{proof}

\end{document}